\documentclass[12pt,a4paper,oneside]{article}
  \usepackage{authblk} 
  \usepackage{url}
  \usepackage{amsthm}
  \usepackage{amsmath}
  \usepackage{amsfonts}
  \usepackage{amssymb}
  \usepackage{amscd}
  \usepackage{authblk}
  \usepackage{bm}
  \usepackage[mathscr]{eucal}
  \usepackage{mathtools} 
  \usepackage{enumitem}
  \usepackage{physics}
  \usepackage{color}
  \usepackage{fancyhdr}
  \usepackage{hyperref}
  \usepackage{marginnote}
  \usepackage[top=2.5cm, bottom=2.7cm,left=2.5cm, right=2.5cm, marginparwidth=1.8cm]{geometry}
  \usepackage{framed}
  \usepackage[noabbrev]{cleveref} 
  \usepackage[T1]{fontenc} 
 \usepackage[utf8]{inputenc}
 \usepackage{blkarray}

\usepackage{relsize} 
  
  \numberwithin{equation}{section}
 
 \allowdisplaybreaks[1] 
  
  \theoremstyle{definition}  
   \newtheorem{defn}{Definition}[section]

   \newtheorem{rmk}[defn]{Remark}

  \theoremstyle{plain}  
   \newtheorem{thm}[defn]{Theorem}
   \newtheorem{lem}[defn]{Lemma}
   \newtheorem{prop}[defn]{Proposition}
   \newtheorem{cor}[defn]{Corollary}

  \theoremstyle{remark} 

   \newcommand{\bell}{\bm{\ell}}
   \newcommand{\bk}{\mathbf{k}}
   \newcommand{\br}{\mathbf{r}}
   \newcommand{\bs}{\mathbf{s}}
   
   \newcommand{\bu}{\mathbf{u}}
   \newcommand{\bv}{\mathbf{v}}

   \newcommand{\CK}{\mathcal{K}}

   \newcommand{\BC}{\mathbb{C}}
   \newcommand{\BR}{\mathbb{R}}

   \newcommand{\BZ}{\mathbb{Z}}

\usepackage{centernot}

 \newcommand{\numberthis}{\refstepcounter{equation}\tag{\theequation}}  
\makeatletter
\newcommand\footnoteref[1]{\protected@xdef\@thefnmark{\ref{#1}}\@footnotemark}
\makeatother

\newcommand{\emailaddress}[1]{\newline{\sf#1}}

\let\tr\relax 
\DeclareMathOperator{\tr}{Tr}

\let\ker\relax 
\DeclareMathOperator{\ker}{Ker}
\DeclareMathOperator{\ran}{Ran}

\DeclareMathOperator{\im}{Im}
\DeclareMathOperator{\supp}{Supp}

\DeclareMathOperator{\spn}{span}

\fontsize{12}{14}

\title{Petz-R\'enyi Relative Entropy of Thermal States and their  Displacements}
\author[1]{George Androulakis}
\author[2]{Tiju Cherian John}

\affil[1]{University of South Carolina, Columbia, South Carolina, USA \emailaddress{giorgis@math.sc.edu}}
\affil[2]{The University of Arizona, Tucson, Arizona, USA \emailaddress{ tijucherian@fulbrightmail.org}}

\begin{document}
\maketitle
\tableofcontents
\begin{abstract}
In this letter, we obtain the precise range of the values of the parameter $\alpha$ such that  Petz-Rényi $\alpha$-relative entropy $D_{\alpha}(\rho||\sigma)$ of two faithful displaced thermal states is finite. More precisely, we prove that, given two displaced thermal states $\rho$ and $\sigma$ with inverse temperature parameters  $r_1, r_2,\dots, r_n$ and $s_1,s_2, \dots, s_n$,  respectively, $0<r_j,s_j<\infty$, for all $j$, we have
 \begin{align*}
      D_{\alpha}(\rho||\sigma)<\infty \Leftrightarrow \alpha < \min \left\{ \frac{s_j}{s_j-r_j}: j \in \{ 1, \ldots , n \} \text{ such that } r_j<s_j \right\},
  \end{align*}  where we adopt the convention that the minimum of an empty set is equal to infinity. This result is particularly useful in the light of operational interpretations of the Petz-Rényi $\alpha$-relative entropy in the regime $\alpha>1 $.  Along the way, we also prove a special case of a conjecture of Seshadreesan, Lami, and Wilde (J. Math. Phys. 59, 072204 (2018)).
\end{abstract}
\section{Introduction}
Quantum state discrimination is an indispensable tool in quantum communication theory. 
Petz-Rényi $\alpha$-relative entropy is a useful tool for discriminating two quantum states \cite{Nussbaum-Szkola-2009, Mosonyi-2009, Ses-Lam-Wil-2018}, and it generalizes the classical Rényi $\alpha$-relative entropy \cite{renyi1961}. For the purpose of this article, for $\alpha\in (0,1) \cup (1,\infty)$, we define Petz-Rényi $\alpha$-relative entropy of two quantum states $\rho$ and $\sigma$, $D_{\alpha}(\rho||\sigma)$, as\begin{align*}
  D_{\alpha}(\rho||\sigma)= \begin{cases}
      \frac{1}{\alpha-1}\log \tr \rho^{\alpha}\sigma^{\alpha-1}, & \text{ when } \supp \rho\subseteq \supp \sigma\\
      \infty, & \text{ otherwise. }
  \end{cases}  
\end{align*} For a detailed account on the definition and basic properties of this notion, one may refer to Petz \cite{petz1986quasientropy} who first defined it, or to \cite[Examples 3.3 (2)]{Androulakis2023-lv} and \cite{Androulakis2023-na}.  

The theory of quantum gaussian states   has been getting more attention recently, in the context of its importance in the continuous variable  quantum information theory \cite{weedbrook-et-al-2012}. Several authors have obtained formulas for the Petz-Rényi $\alpha$-relative entropy, $D_{\alpha}(\rho||\sigma)$ between gaussian states $\rho$ and $\sigma$  \cite{Mosonyi-2009,Ses-Lam-Wil-2018, Par2021a}.  For such states,  it is known that  \[ D_{\alpha}(\rho ||\sigma)<\infty, \quad \forall \alpha \in (0,1].\] However, the precise range of $\alpha$ beyond $1$, where the Petz-Rényi $\alpha$-relative entropy is finite for gaussian states is not known.

In this article, we study the regime $\alpha>1$ of the Petz-Rényi $\alpha$-relative entropy for gaussian states. In this case, Petz-Rényi $\alpha$-relative entropy has an operational interpretation in connection with optimal quantum encoding; see \cite[Theorem 7]{bellomo2017lossless}, which we describe now. Consider a quantum source that emits pure states $\rho_i$ with certain probabilities
$p_i$. Thus, the average state emitted by the quantum source is equal to the mixed state $\rho=\sum_i p_i \rho_i$. The $t$-exponential length of a codeword is defined in \cite[Definition 6]{bellomo2017lossless}. If we use the optimal 
encoding for $\rho$, then the $t$-exponential length $\ell_t^{\text{opt}}$ of the encoding of the average state emitted by the source satisfies 
$$
D_{\frac{1}{1+t}} (\rho) \leq \ell_t^{\text{opt}} \leq D_{\frac{1}{1+t}}(\rho) +1,
$$ where for $\alpha>0$, $D_{\alpha}(\rho)$ is the Petz-Rényi $\alpha$-entropy of $\rho$, (see 
\cite[Theorem 6]{bellomo2017lossless}). 
On the other hand, if we use the quantum version of the Shannon encoding in order to encode the average state that is emitted by the 
quantum source, and moreover, if we construct that quantum Shannon encoding based on a mixed state $\tau$ which is potentially different than $\rho$, then 
the $t$-exponential length $\ell_t^\text{Sh}$ of the encoding of the average state emitted by the quantum source satisfies 
\begin{equation} \label{eq:bellomo-interpretation}
D_{\frac{1}{1+t}}(\rho)+D_{1+t}\left(\rho_t \| \tau_t\right) \leq \ell_t^{\text{Sh}}<D_{\frac{1}{1+t}}(\rho)+D_{1+t}\left(\rho_t \| \tau_t\right)+1
\end{equation}
where $\rho_t:=(\tr \rho^{\frac{1}{1+t}})^{-1}\rho^{\frac{1}{1+t}}$ and also $\tau_t$ is defined similarly, (see \cite[Theorem 7]{bellomo2017lossless}).  Equation~\eqref{eq:bellomo-interpretation} gives an operational interpretation of the 
Petz-Renyi $\alpha$-relative entropy for $\alpha >1$. Notice that if the states $\rho$ and $\tau$ are thermal, then the states 
$\rho_t$ and $\tau_t$ are  also thermal.

Furthermore, the Petz-Rényi $\alpha$-relative entropy is known to satisfy the data processing inequality in the range of $1<\alpha\leq 2$ , see \cite[Theorem 1.1, part 2]{Zhang2020-gk} and \cite{camilo-landi-gabriel-2019}. Moreover, the Petz-Rényi $2$-relative entropy is particularly useful for Gaussian states (see, for instance, \cite{adesso-girolami-serafini-2012, Camilo2019-zl} and the Introduction in \cite{Iosue2023-qe}).
The reason that the Petz-Rényi 2-relative entropy is preferred occasionally than the Umegaki relative entropy for state distinguishability is because it is the largest among the Petz-Rényi relative entropies that satisfy the data processing
inequality.

Thus, for the reasons mentioned in the previous two paragraphs, it is important to know the range of $\alpha$'s beyond $1$ for which the Petz-Rényi $\alpha$-relative entropy is finite. In particular, it is also important to know the answer to this question when the states are thermal. In this article, we explicitly find this range for thermal states and their displacements (Theorems \ref{thm:not-necessarily-faithful} and \ref{thm:main-faithful}). In addition, our investigations led to the settlement of a special case of a recent conjecture by Seshdreesan, Lami, and Wilde \cite{Ses-Lam-Wil-2018}.  The authors of \cite{Ses-Lam-Wil-2018} prove that given two faithful gaussian states $\rho$ and $\sigma$, and  $\alpha>1$, $D_\alpha(\rho||\sigma)$ is finite  if the covariance matrix of the gaussian state $\frac{\rho^\alpha}{\tr \rho^\alpha}$ is strictly less  than that of the gaussian state $\frac{\sigma^{\alpha-1}}{\tr \sigma^{\alpha-1}}$. They conjectured that the converse of this statement is also true. In this article, we verify this conjecture for displaced faithful thermal states (Theorem \ref{thm:converse-SLW}). Unfortunately, the methods adopted in this article produce difficult expressions when we consider more general gaussian states other than displaced thermal states. Hence, in this article, we restrict ourselves to thermal states and their displacements.

There are two main mathematical tools that we use in obtaining our results in this article. The first is the well-known spectral decomposition of thermal states. The second tool is an analysis of the matrix entries of the Weyl unitary operators (a.k.a. displacement operators) in the particle basis using a theorem of Fejér from classical analysis (Propositions \ref{prop:weyl-leguerre} and \ref{weyl-diagonal} in this article).

The structure of this article is as follows. In Section \ref{sec:preliminaries}, we discuss the spectral decomposition of multimode thermal states after providing the necessary mathematical background. In Section \ref{sec:main-results}, we include the main results of this article where Subsection \ref{sec:thermal} discusses the case of thermal states, and Subsection \ref{sec:displaced-thermal} provides results for displaced faithful thermal states.
\section{Preliminaries}\label{sec:preliminaries}

\subsection{Thermal states}
  Let $a$ and $a^{\dagger}$ denote the annihilation and creation operators of a single-mode system.  These operators enjoy the following properties (refer \cite{Par12, Par10} for details).
\begin{enumerate}
\item The  matrix representation of the  annihilation operator $a$ and the creation operator $a^{\dagger}$ in the particle basis $\{\ket{0},  \ket{1}, \ket{2}, \dots\}$ is given as,
  \begin{eqnarray}
    \label{eq:37}
a = 
    \begin{bmatrix}
      0 &1 &0&0&\cdots\\
 0&0&\sqrt{2}&0& \cdots\\ 
0&0&0&\sqrt{3}&\cdots \\
0&0&0&0&\cdots \\
\cdots&\cdots&\cdots&\cdots&\cdots
    \end{bmatrix}; \phantom{....}a^{\dagger} =
    \begin{bmatrix}
      0 &0 &0&0&\cdots\\
 1&0&0&0& \cdots\\ 
0&\sqrt{2}&0&0&\cdots \\
0&0&\sqrt{3}&0&\cdots\\
\cdots&\cdots&\cdots&\cdots&\cdots
    \end{bmatrix}.
  \end{eqnarray}
In other words we have, \begin{align}
\label{eq:6.3}
a = \sum\limits_{j=1}^{\infty}\sqrt{j}\ketbra{(j-1)}{j},&\phantom{..}  a^{\dagger} = \sum\limits_{j=1}^{\infty}\sqrt{j}\ketbra{j}{(j-1)}.
\end{align}

\item The operator $a^{\dagger}a$ defined on the finite particle domain and extended as a selfadjoint operator to its maximal domain is called the \emph{number \index{number operator}operator}. The number operator has a diagonal representation in the particle basis given by
\begin{equation}
\label{eq:38}
a^{\dagger}a \ket{j} = j\ket{j}, j\in\BZ_{\geq 0},
\end{equation}
where $\BZ_{\geq 0}$  denotes the set of all nonnegative integers.  Equivalently, 
\begin{align*}
a^{\dagger}a =  \sum\limits_{j=0}^{\infty} j \ketbra{j}{j}
\end{align*}
with the  matrix representation  in the particle basis  given as, \begin{align*}
    a^{\dagger}a =   \begin{bmatrix}
     0&&&& \\
      &1 &&\mbox{\huge 0}&\\
& &{2}&& \\ 
&\mbox{\huge 0}&&{3}& \\
&&&&\ddots \\
    \end{bmatrix}.
\end{align*}
Thus, $a^{\dagger}a$ has a complete set of orthonormal eigenvectors, $\{\ket{0}, \ket{1}, \ket{2}, \dots\}$.
\item For $0<s\leq \infty$, define the bounded operator $e^{-s a^{\dagger}a}$ to be the operator that satisfies  
\begin{equation}
\label{eq:6.4} 
\begin{cases}
e^{-s a^{\dagger}a} \ket{j}= e^{-sj}\ket{j}, \textnormal{ for } j=0, 1, 2, \dots, & \textnormal{when }0<s<\infty ;\\
\phantom{..}e^{-s a^{\dagger}a}\phantom{...} = \ketbra{0}, & \textnormal{when }s=\infty.
\end{cases}
\end{equation}
\item The trace of $e^{-s a^{\dagger}a}$ can be computed as the sum of eigenvalues. Hence 
\begin{align*}
\tr (e^{-s a^{\dagger}a}) &= 1+e^{-s}+ (e^{-s})^2+(e^{-s})^3+\cdots\\
&= \frac{1}{1-e^{-s}}.
\end{align*}
Therefore, the operator $(1-e^{-s})e^{-sa^\dagger a}$ is a state (i.e., a positive operator with unit trace)
\end{enumerate}
\begin{defn}[Thermal State]
 Let $0<s\leq \infty$, define
\begin{equation}\label{eq:thermal-s}
    \gamma(s)= 
    \begin{cases}
    (1-e^{-s})e^{-sa^{\dagger}a}, & 0<s<\infty\\
    \ketbra{0}, & s=\infty.
    \end{cases}
\end{equation}
Then $\gamma(s)$ is called a $1$-mode thermal state with inverse temperature $s$. If $0<s_j\leq \infty$, for $j=1,2, \dots,n$, then the product state \begin{align}\label{eq:n-mode-thermal}
    \gamma(\bs) = \otimes_{j=1}^n\gamma(s_j).
\end{align}
is called an $n$-mode thermal state with inverse temperature $s_j$ on the $j$-th mode.
\end{defn}
\begin{rmk}\label{rmk:thermal} \begin{enumerate}
   \item \label{rmk-thermal-covariance-1}
  It may be recalled \cite{Par10}, that $\gamma(\bs)$ is an $n$-mode (quantum) gaussian state with mean zero and covariance matrix  as the diagonal matrix \begin{align}
   V = \frac{1}{2}\bmqty{\dmat{\coth\frac{s_1}{2},\ddots,\coth\frac{s_n}{2}}} \otimes I_2,
\end{align} 
where $\coth \infty$ is interpreted as $1$.
   \item \label{rmk-thermal-covariance-2}
   Note that any positive power of the one mode thermal state $\gamma(s)$ is also a trace class operator. Furthermore, we have \[\frac{\gamma(s)^\alpha}{\tr \gamma(s)^\alpha} = \gamma(\alpha s),\quad \forall \alpha >0.\] 
   By taking tensor products we have for $\bs=(s_1,s_2,\dots,s_n)$, with $0<s_j\leq \infty$, \[\frac{\gamma(\bs)^\alpha}{\tr \gamma(\bs)^\alpha} =\gamma(\alpha \bs), \]
   where $\alpha\bs = (\alpha s_1,\alpha s_2,\dots,\alpha s_n). $ Therefore, the covariance matrix associated with the state $\frac{\gamma(\bs)^\alpha}{\tr \gamma(\bs)^\alpha} $ is \begin{equation}
       \label{eq:rho-alpha-cov}
       V(\alpha)  = \frac{1}{2}\bmqty{\dmat{\coth\frac{\alpha s_1}{2},\ddots,\coth\frac{\alpha s_n}{2}}} \otimes I_2,
   \end{equation}
   where $\coth \infty$ is interpreted as $1$.
\end{enumerate}
\end{rmk}
\subsubsection{Spectral decomposition of Thermal States}
Due to \eqref{eq:6.4} and the  definition of a thermal state, the spectral decomposition of the one mode thermal state $\gamma(s)$ is given by 
\begin{equation}\label{eq:1-mode-spectral}
    \gamma(s) = \begin{cases}
      (1-e^{-s})\sum\limits_{k\in \BZ_{\geq0}}e^{-ks}\ketbra{k}{k}, & 0<s<\infty;\\
\phantom{................} \ketbra{0}, & s=\infty.
    \end{cases}
\end{equation}
Since the $n$-mode thermal states are tensor product of $1$-mode thermal states, the spectral decomposition of these states are straight-forward. 

We adopt the following conventions to make the exposition simpler: \[0\cdot \infty = 0, \quad \quad e^{-\infty} = 0.\] 
With these conventions, we have 
\begin{align} \label{beq:1-mode-spectrum}
\gamma(s)\ket{\ell} = (1-e^{-s})e^{-\ell s}\ket{\ell}, \quad \forall \ell\in \BZ_{\geq 0}, \text{ and } 0<s\leq \infty .
\end{align} To express the eigenvalues of $\gamma(\bs)= \otimes_{j=1}^n\gamma(s_j)$, 
 we introduce the following notation.
For $\bm{\ell}\in \BZ_{\geq 0}^{n}$, define $\lambda(\bm{\ell},\bs)$ by
\begin{equation}\label{beq:5}
\lambda(\bm{\ell},\bs)=
   \left(\prod_{j=1}^{n}(1-e^{-s_j})\right)e^{-\bm{\ell}\cdot\bs}
\end{equation}
where $\bm{\ell}\cdot\bs = \sum_{j=1}^n \ell_j s_j$. 
Therefore, the action of $\gamma(\bs)=\otimes_{j=1}^n \gamma(s_j)$ on the particle
basis vector \[\ket{\bm{\ell}}:=\ket{\ell_1}\otimes \ket{\ell_2}\otimes\cdots\otimes\ket{\ell_n}\]  is
given by
\begin{align*}
\gamma(\bs)\ket{\bm{\ell}}=\otimes_{j=1}^n \gamma(s_j)\ket{\bm{\ell}} =\lambda(\bell,\bs)\ket{\bell} &=
 \left(\prod_{j=1}^{n}(1-e^{-s_j})\right)e^{-\bm{\ell}\cdot\bs}\ket{\bm{\ell}}.
\end{align*}
In particular, the kernel of $\gamma(\bs)$ is the closure of the span of those particle basis vectors $\ket{\bell}$ where the coordinates of $\bell\in \BZ_{\geq 0}^n$   corresponding to $s_j=\infty$ are non-zero. We have 
\begin{align}
    \ker (\gamma(\bs))&=\spn \left\{\ket{\bell}: \ell= (\ell_1,\dots,\ell_n) \in \BZ_{\geq 0}, s_j=\infty \Rightarrow \ell_j>0 \right\} \label{eq:ker-gamma-s}\\
    \ran (\gamma(\bs))&= \spn \left\{\ket{\bell}: \ell= (\ell_1,\dots,\ell_n) \in \BZ_{\geq 0}, s_j=\infty \Rightarrow \ell_j=0 \right\} \label{eq:ran-gamma-s}.
\end{align}
The spectral decomposition of $\gamma(\bs)$ is given by \begin{align}\label{beq:n-mode-spectral}
    \begin{split}
        \gamma(\bs) &= \sum\limits_{\bm{\ell}\in\BZ_{\geq 0}^{n} } \lambda(\bm{\ell}, \bs) \ketbra{\bm{\ell}}.
    \end{split}
\end{align}

\subsection{Weyl Unitary Operators and Displaced Thermal States} The thermal states that we described in the previous sections are states with zero mean. They can be transformed to states with any non-zero mean vector (without altering the covariance matrix) by conjugating them with the \textit{Weyl unitary operators}. Now we define the Weyl unitary operators associated with a $n$-mode system. \begin{defn} The \emph{exponential vector} $\ket{e(\bv)}$ at $\bv\in \BC^n$ is defined as \[\ket{e(\bv)} = \sum_{k\in \BZ_{\geq 0}}^{\infty}\frac{\bv^{\otimes^k}}{\sqrt{k!}}.\] Then the \emph{Weyl unitary operator} or the \emph{displacement operator} $W(\bu)$ at $\bu \in \BC^n$ is defined as
\begin{align}\label{eq:weyl-exp}
    W(\bu)\ket{e(\bv)} = \exp{-\frac12 \norm{\bu}^2- \braket{\bu}{\bv}}\ket{e(\bu +\bv)}, \quad \forall \bu, \bv \in \BC^n.
\end{align}
  A state of the form \[W(\bu)\gamma(\br)W(\bu)^{\dagger}\]
is called a \emph{displaced thermal state}.
\end{defn}
The following properties of the Weyl unitary operators and displaced thermal states are important for our exposition.
\begin{enumerate}
    \item Weyl unitary operators enjoy the following properties which makes the map $\bu\mapsto W(\bu)$  a projective unitary representation of the additive group $\BC^n$, \begin{align}\label{eq:weyl-representation}\begin{split}
    W(\bu)^{\dagger} &= W(-\bu) \\
        W(\bu)W(\bv) &= \exp(-i \im \left\langle \bu|\bv \right\rangle)W(\bu+\bv)\left.\right\}, \quad \forall \bu, \bv \in \BC^n.
        \end{split} 
    \end{align}
    \item The covariance matrix of the thermal state $\gamma(\bs)$, remains invariant under the conjugation with Weyl unitary operator. In other words, any displaced thermal state has the same covariance matrix as that of the corresponding thermal state. 
    \item The spectral decomposition of the displaced thermal state $W(\bu)\gamma(\br)W(\bu)^{\dagger}$ is given by 
\begin{align}\label{beq:displaced-thermal-spectral}
    \begin{split}
W(\bu)\gamma(\br)W(\bu)^{\dagger} &=\sum\limits_{{\bk}\in\BZ_{\geq 0}^{n} } \lambda({\bk}, \br) \ketbra{W(\bu){\bk}}.    \end{split}
    \end{align}
\end{enumerate}
\subsection{Formula for Petz-Rényi Relative Entropy of Thermal States and their Displacements}
 Let $\rho$ and $\sigma$ be any two states (not necessarily thermal states) on a Hilbert space $\CK$  with spectral decomposition \begin{align}\label{eq:spectral-rho-sigma}\begin{split}
    \rho &= \sum_{i\in \mathcal{I}}r_i \ketbra{u_i}, \quad r_i\geq 0,\quad \sum_{i\in \mathcal{I}} r_i = 1,\quad \{u_i\}_{i\in \mathcal{I}} \text{ is an orthonormal basis of }  \mathcal{K};\\
    \sigma &= \sum_{j\in \mathcal{I}}s_j \ketbra{v_j}, \quad s_j\geq0, \quad \sum_{j\in \mathcal{I}} s_j = 1, \quad \{v_j\}_{j\in \mathcal{I}}\text{ is an orthonormal basis of } \mathcal{K}.
    \end{split}
\end{align}
 Then by   \cite[Theorem 2.6]{Androulakis2023-na}, for $\alpha \in (0,1)\cup (1,\infty)$, the Petz-R\'enyi $\alpha$-relative entropy is given by
 \begin{align}\label{eq:alpha-relative-entropy}
   D_{\alpha}(\rho||\sigma)&= \frac{1}{\alpha-1} \log \sum\limits_{i,j\in \mathcal{I}} r_i^\alpha s_j^{1-\alpha}\abs{\braket{u_i}{v_j}}^2, \quad \forall \alpha \in (0,1) \cup (1,\infty),
\end{align}
where for $\alpha>1$, we adopt the conventions $0^{1-\alpha} = \infty$ and $0\cdot \infty = 0$. In particular, if $u_i = v_i$ for all $i$, that is, if both the states $\rho$ and $\sigma$ are diagonalized in the same orthonormal basis, then $\abs{\braket{u_i}{v_j}}^2 = \delta_{ij}$ and we have \begin{align}\label{eq:alpha-relative-entropy-1}
   D_{\alpha}(\rho||\sigma)&= \frac{1}{\alpha-1} \log \sum\limits_{i} r_i^\alpha s_i^{1-\alpha}, \quad \forall \alpha \in (0,1) \cup (1,\infty).
\end{align} 
\begin{rmk}\label{rmk:support-condition}When $\sigma$ has a non-trivial kernel, that is, $s_j=0$ for some $j$, then for $\alpha>1$, the term $s_j^{1-\alpha}$ is infinity by our convention. In this case, the corresponding term $r_i^\alpha s_j^{1-\alpha}\abs{\braket{u_i}{v_j}}^2$ in the sum in \eqref{eq:alpha-relative-entropy} is infinity (implying $D_{\alpha}(\rho||\sigma) = \infty$) unless $\braket{u_i}{v_j}= 0$ or $r_i=0$ for every $i$. Thus, for $\alpha >1$, the condition \begin{equation}\label{eq:support-condition}
    \forall j \text{ such that } s_j= 0, \text{ either } \braket{u_i}{v_j}= 0 \text{ or } r_i=0 , \quad \forall i 
\end{equation}
is a necessary condition for the finiteness of $D_\alpha(\rho||\sigma)$. It is easy to see that the condition \eqref{eq:support-condition} is equivalent to the condition \[\supp \rho \subseteq \supp \sigma\] (for a proof, the reader may refer to \cite[Proof of Lemma 3.3, Appendix A.2]{Androulakis-John-2022a} or \cite[Lemma 3.11]{Androulakis2023-lv}). Thus, if the condition $\supp \rho \subseteq \supp \sigma$ is satisfied, then the formula \ref{eq:alpha-relative-entropy} reduces to  \begin{align}\label{eq:alpha-relative-entropy-support-cond}
   D_{\alpha}(\rho||\sigma)&= \frac{1}{\alpha-1} \log \sum\limits_{\{i,j\in \mathcal{I}| r_i\neq 0, s_j \neq 0\}} r_i^\alpha s_j^{1-\alpha}\abs{\braket{u_i}{v_j}}^2, \quad \forall \alpha \in (0,1) \cup (1,\infty),
\end{align}
\end{rmk}
In the case of thermal states,  $\gamma(\br) =\otimes_{j=1}^n \gamma(r_j)$  and $\gamma(\bs)=\otimes_{j=1}^n\gamma(s_j)$, both these states have a diagonal representation with respect to the particle basis. Therefore, we have the following:
    \begin{align}\label{eq:12}
&D_\alpha(\gamma(\br) \lvert\lvert \gamma(\bs)) =\frac{1}{\alpha-1}\log \sum\limits_{\bk\in \BZ_{\geq 0}^n}\lambda(\bk,\br)^\alpha\lambda(\bk,\bs)^{(1-\alpha)},  
\end{align} where 
$\lambda(\bk,\br)$ and $\lambda(\bk,\bs)$ are as in (\ref{beq:5}).
\begin{prop}\label{prop:d-alpha-displaced-thermal}
     Let $\rho$ and $\sigma$ be $n$-mode displaced thermal states such that \begin{align*}
        \rho&= W(\bu_1)\gamma(\br)W(\bu_1)^{\dagger},\\
        \sigma& = W(\bu_2)\gamma(\bs)W(\bu_2)^{\dagger}, 
    \end{align*}
where $\bu_1 , \bu_2 \in \BC^{n}$. Then \[D_{\alpha}(\rho||\sigma) = \frac{1}{\alpha-1}\log\sum\limits_{\bk, \bm{\ell}\in \BZ_{\geq 0}^n }\lambda(\bk,\br)^\alpha\lambda(\bm{\ell},\bs)^{(1-\alpha)}\abs{\braket{W(\bu)\bk}{\bm{\ell}}}^2,\]
where $\bu = \bu_1-\bu_2$.
\end{prop}
\begin{proof}Note that the eigenvalues of $\rho$ are same as that of $\gamma(\br)$ and the eigenvector corresponding to the eigenvalue $\lambda(\bk, \br)$ is now $W(\bu_1)\ket{\bk}$. Similarly, the eigenvector corresponding to the eigenvalue $\lambda(\bell, \bs)$ of $\sigma$ is $W(\bu_2)\ket{\bell}$.
Now by \eqref{eq:alpha-relative-entropy},  and \eqref{beq:displaced-thermal-spectral}, \[D_{\alpha}(\rho||\sigma) = \frac{1}{\alpha-1}\log\sum\limits_{\bk, \bm{\ell}\in \BZ_{\geq 0}^n }\lambda(\bk,\br)^\alpha\lambda(\bm{\ell},\bs)^{(1-\alpha)}\abs{\braket{W(\bu_1)\bk}{W(\bu_2)\bm{\ell}}}^2, \]  which together with \eqref{eq:weyl-representation} completes the proof. 
\end{proof}
\section{Main Results}\label{sec:main-results}
\subsection{Thermal States}\label{sec:thermal}
Given two gaussian states $\rho$ and $\sigma$,  there exist thermal states $\gamma(\br)$ and $\gamma(\bs)$ and gaussian unitary operators $U$ and $V$ such that  \[\rho = U\gamma(\br)U^{\dagger}, \quad \text{and}\quad V\gamma(\bs)V^{\dagger},\]\cite{Par10, BhJoSr18}.  It is easy to see from \eqref{beq:5}  that  the eigenvalues $\lambda(\bk,\br)$ of $\gamma(\br)$ are exponential in $\bk$
and tend to zero for large $\bk$, therefore from \eqref{eq:alpha-relative-entropy} the Petz-Rényi $\alpha$-relative entropy of two gaussian states $\rho$ and $\sigma$ satisfies  \begin{equation}\label{eq:d-alpha-finite}
    D_{\alpha}(\rho||\sigma) <\infty, \quad \forall \alpha\in (0,1],
\end{equation} (see also \cite{Mosonyi-2009, Ses-Lam-Wil-2018}).  But the precise range of $\alpha$ beyond $1$ for which  $D_{\alpha}(\rho||\sigma) <\infty$ is not known. Our next theorem settles this question for thermal states. 
It will be convenient to state our results by introducing the following function $g: (0, \infty] \times (0,\infty] \to (1, \infty]$ given by
\begin{align}\label{eq:g}
    g(r,s)= \left\{ \begin{array}{ll} \frac{s}{s-r}, & \text{if } 0<r<s<\infty\\ \infty, & \text{otherwise.} \end{array}\right.  
\end{align}

\begin{thm}\label{thm:not-necessarily-faithful} Let  $\gamma(\br) =\otimes_{j=1}^n \gamma(r_j)$  and $\gamma(\bs)=\otimes_{j=1}^n\gamma(s_j)$
 be thermal states, $0<r_j\leq\infty$, $0<s_j\leq\infty$, such that $\supp \gamma(\br)\subseteq \supp\gamma(\bs)$.
Then, 
$$
D_\alpha (\gamma(\br) || \gamma(\bs) ) < \infty \Leftrightarrow \alpha < \min_{1\leq j \leq n} g(r_j,s_j).
$$
\end{thm}
\begin{proof} Because of \eqref{eq:d-alpha-finite}, assume $\alpha\geq 1$.
From the definition of
$\lambda(\bk,\br)$ and $\lambda(\bk,\bs)$ in (\ref{beq:5}), and \eqref{eq:12},
$D_{\alpha}(\rho||\sigma)$ is finite if and only if 
\begin{equation}\label{eq:sum}
    \sum\limits_{\bk\in \BZ_{\geq 0}^n} \lambda(\bk, \br)^{\alpha} \lambda(\bk,\bs)^{(1-\alpha)}<\infty.
\end{equation} 
Notice that if $s_j=\infty$ for some $j$, then the condition $\supp(\gamma(\br))\subseteq \supp(\gamma(\bs))$ implies that $r_j=\infty$ as well. Thus, $\lambda(\bk, \bs) = 0$ for some $\bk$ implies $\lambda(\bk,\br)=0$. Therefore, if $s_j=\infty$ for some $j$, then the terms $\lambda(\bk, \br)^{\alpha} \lambda(\bk,\bs)^{(1-\alpha)}$ become $0\cdot \infty = 0$, for all $\bk=(k_1,\dots, k_n)$ with $k_j>0$, and they do not contribute to the expression \eqref{eq:sum}.

Also, note that if there is a $j$ such that $r_j=\infty$, then $\lambda(\bk,\br) = 0$ for all $\bk =(k_1, \ldots , k_n)$ with $k_j>0$ and as before those terms do not contribute non-zero terms to the expression \eqref{eq:sum}. 

To summarize the previous two paragraphs, if $s_j=\infty$ or $r_j=\infty$ for some $j$, the only terms that possibly have nonzero contributions to the sum in \eqref{eq:sum} correspond to those $\bk$'s where $k_j=0$. 
On the other hand, notice that \[ \bk \cdot \br = \sum_{ \{ i: k_i>0 \} } k_i r_i \text{ and } 
\bk \cdot \bs = \sum_{ \{ i: k_i>0 \} } k_i s_i ,\]
because we do not need to write the $0 \cdot \infty $ terms (i.e. the $j$'s such that $r_j=\infty$ or $s_j=\infty$ and $k_j=0$) in the dot products $\bk \cdot \br$ and $\bk \cdot \bs$. Hence, if 
\[ \{ j: r_j,s_j<\infty \} =
\{j_1,j_2,\dots, j_p\},\] 
then  \begin{align*}
  \lambda(\bk, \br)^{\alpha} \lambda(\bk,\bs)^{(1-\alpha)}& =  C\exp{-\alpha \bk\cdot \br} \exp{-(1-\alpha) \bk\cdot \bs}\\
  &=C\exp{-\alpha\left(\sum\limits_{i=1}^pk_ir_{j_i}\right)}\exp{-(1-\alpha)\left(\sum\limits_{i=1}^pk_is_{j_i}\right)},
\end{align*}
where the constant term $C$ depends only on $\br, \bs,$ and $\alpha$.
Now we have
\begin{align*}
&D_\alpha (\gamma(\br) ||\gamma(\bs))<\infty \\
  \Longleftrightarrow& \sum\limits_{k_{1},\dots, k_p \in \BZ_{\geq
                       0}}\exp{-\alpha\left(\sum\limits_{i=1}^pk_ir_{j_i}\right)}\exp{-(1-\alpha)\left(\sum\limits_{i=1}^pk_is_{j_i}\right)}<\infty
  \\
   \Longleftrightarrow& \prod_{i=1}^p
                        \left(\sum\limits_{k_{i} \in \BZ_{\geq
                        0}}\exp{-\left(\alpha
                        r_{j_i}+(1-\alpha)s_{j_i}\right)k_i}\right)<\infty\\
  \Longleftrightarrow&  \sum\limits_{k_{i} \in \BZ_{\geq
                        0}}\exp{-\left(\alpha
                        r_{j_i}+(1-\alpha)s_{j_i}\right)k_i}<\infty, \quad \forall  i\\
  \Longleftrightarrow&  \alpha
                       r_{j_i}+(1-\alpha)s_{j_i}>0, \quad \forall  i\\
 \Longleftrightarrow& \alpha
                     (s_{j_i}- r_{j_i})< s_{j_i}, \quad \forall
                      i.\numberthis \label{eq:21}
\end{align*}
Now there are two cases:  either $r_{j_{i}}\geq s_{j_i}, \forall i,$  or  $r_{j_i}< s_{j_i}$  for some $i$. If $r_{j_{i}}\geq s_{j_i}, \forall i,$ 
then (\ref{eq:21}) is automatically satisfied for all
$\alpha>0.$  Since $g(r_j,s_j)=\infty$ for all $j$, the statement of the theorem is proved in this case.  If there exists a $j_i$ such that $r_{j_i}< s_{j_i}$  
then (\ref{eq:21}) is satisfied for some
$\alpha>0$ if and only if $\alpha< \min_{i}\left\{\frac{s_{j_i}}{s_{j_i}-r_{j_i}}\right\} = \min_{j}g(r_j,s_j)$ and the statement of the theorem is also verified in this.
\end{proof}
The next remark rephrases the above theorem in more detail.

\begin{rmk}\label{rmk:main-thm}
       \begin{enumerate}
\item\label{item:13}
  If  $r_j\geq s_j$ for all $j$
   such that $r_j$ and $s_j$ are finite,
  then
\[D_{\alpha}(\gamma(\br)||\gamma(\bs))<\infty, \quad \forall \alpha >0.\] The condition $r_j\geq s_j$  means that the temperatures $\frac{1}{r_j}\leq \frac{1}{s_j}$ for all modes that are in nonzero temperature. This is one definite way to obtain a finite Petz-Rényi $\alpha$-entropy for all $\alpha$.
\item \label{item:15-rmk} If $\{  j : r_j<s_j<\infty \} \not = \emptyset$,  then 
\begin{align}
\label{eq:16}
D_{\alpha}(\gamma(\br)||\gamma(\bs))<\infty \Longleftrightarrow \alpha< \min\left\{\frac{s_j}{s_j-r_j}: r_j<s_j<\infty\right\}.
\end{align}
\end{enumerate}
\end{rmk}
\begin{rmk} 
 It is an interesting special case  to note that if $r_i$'s and $s_i$'s satisfy
two properties \begin{enumerate}
    \item $r_j<\infty \Rightarrow s_j = \infty$, and
    \item  $s_k < \infty \Rightarrow r_k = \infty$,
\end{enumerate}   then
\begin{align*}
D_{\alpha}(\gamma(\br)||\gamma(\bs))<\infty, \quad \forall \alpha,
\end{align*} 
i.e., if we prepare $\rho$ and $\sigma$ such that for each mode of $\rho$ with finite non zero temperature prepare the corresponding mode of $\sigma$ to be vacuum and vice versa, then this ensures that Petz-R\'enyi $\alpha$-relative entropy is finite for all $\alpha$. 

\end{rmk}


 It is well known that the covariance matrix $V\in M_{2n}(\BR)$ of a $n$-mode quantum state satisfies the inequality,\[V+iJ\geq 0,\] where $J=\bmqty{0&I_n\\-I_n&0}$ is the matrix of the standard symplectic form on $\BR^{2n}$.    Now let $\rho$ and $\sigma$ be gaussian states with covariance matrices $V_{\rho}$ and $V_{\sigma}$ respectively. For $\alpha>1$, let $V_{\sigma(\alpha-1)}$ and $V_{\rho(\alpha)}$ denote the covariance matrices of the gaussian states $\frac{\sigma^{\alpha-1}}{\tr \sigma^{\alpha-1}} $ and $\frac{\rho^{\alpha}}{\tr \rho^{\alpha}}$, respectively. The assumption of the strict inequality between the covariance matrices that appears in the next lemma is vacuous if the state $\sigma$ is not faithful. Moreover, the algebraic calculations that appear in the proof of the lemma become problematic if some of the inverse temperature parameters of the states $\rho$ or $\sigma$
 are equal to infinity. Hence, in the next lemma we assume that the states  
 are faithful.  
\begin{lem}\label{lem:slw}
    Let $\rho =\otimes_{i=1}^n \gamma(r_i)$  and $\sigma=\otimes_{j=1}^n\gamma(s_j)$  be $n$-mode faithful thermal states, i.e., $ 0<r_i,s_i<\infty$ for all $i$. For $\alpha>1$, let $V_{\sigma(\alpha-1)}$ and $V_{\rho(\alpha)}$ denote the covariance matrices of the gaussian states $\frac{\sigma^{\alpha-1}}{\tr \sigma^{\alpha-1}} $ and $\frac{\rho^{\alpha}}{\tr \rho^{\alpha}}$, respectively. Then we have the equivalence:
 $$
V_{\sigma(\alpha-1)}>V_{\rho(\alpha)}\Longleftrightarrow \alpha < \min_{1\leq j\leq n}g(r_j,s_j).$$
\end{lem}
\begin{proof}
    By \ref{rmk-thermal-covariance-2} in Remark \ref{rmk:thermal} we have \begin{align*}
  V_{\sigma(\alpha-1)}  &= \frac{1}{2}\bmqty{\dmat{\coth\frac{(\alpha-1) s_1}{2},\ddots,\coth\frac{(\alpha-1) s_n}{2}}} \otimes I_2, \text{ and}\\    V_{\rho(\alpha)}  &= \frac{1}{2}\bmqty{\dmat{\coth\frac{\alpha r_1}{2},\ddots,\coth\frac{\alpha r_n}{2}}} \otimes I_2.
\end{align*} Since the function $\coth$ is monotonically decreasing on $(0,\infty)$, we have
    \begin{align}\label{eq:proof}
\begin{split}
   &V_{\sigma(\alpha-1)}>V_{\rho(\alpha)}\\
  &\Leftrightarrow \oplus_j \frac{1}{2}\coth s_j(\alpha-1) >\oplus_j \frac{1}{2}\coth \alpha r_j\\
  &\Leftrightarrow  s_j(\alpha-1) < r_j\alpha, \quad \forall j\\
 & \Leftrightarrow  (s_j-r_j)\alpha <s_j, \quad \forall j
 \end{split}
 \end{align}
 Note that the last displayed condition is satisfied for all $\alpha$ if $r_j \geq s_j$ for all $j \in \{ 1, \ldots , n \}$. 
 Otherwise, if there are some $j$'s such that $r_j < s_j$, then the last displayed condition is equivalent to  
$$
\alpha < \min \left\{ \frac{s_j}{s_j-r_j}: j \in \{ 1, \ldots , n \} \text{ such that } r_j<s_j \right\} .
$$
\end{proof}
\begin{rmk} \label{rmk:V_inequality}
Under the assumptions of Lemma~\ref{lem:slw}, its conclusion can be reformulated more explicitly as follows:
$$
V_{\sigma(\alpha-1)}>V_{\rho(\alpha)}\Longleftrightarrow \alpha < \min \left\{ \frac{s_j}{s_j-r_j}: j \in \{ 1, \ldots , n \} \text{ such that } r_j<s_j \right\},
$$
where we adopt the convention that the minimum of an empty set is equal to infinity.
\end{rmk}
Seshadreesan, Lami and Wilde \cite[Theorem 19]{Ses-Lam-Wil-2018}
proved that, \[V_{\sigma(\alpha-1)}>V_{\rho(\alpha)}\Rightarrow D_{\alpha}(\rho||\sigma)<\infty.\] Furthermore, it was conjectured that \cite[Section VIII]{Ses-Lam-Wil-2018} the above implication may also hold in the other direction. Here we prove this conjecture for thermal states and their generalized displacements (i.e., conjugation with Weyl unitary operators). Before proving the general case, let us first consider the special case where the states $\rho$ and $\sigma$ are thermal states. 
\begin{prop}\label{prop:finite-entropy-thermal}
  Let $\rho =\otimes_{i=1}^n \gamma(r_i)$  and $\sigma=\otimes_{j=1}^n\gamma(s_j)$  be $n$-mode faithful thermal states, i.e., $ 0<r_i,s_i<\infty$ for all $i$. For $\alpha>1$, let $V_{\sigma(\alpha-1)}$ and $V_{\rho(\alpha)}$ denote the covariance matrices of the gaussian states $\frac{\sigma^{\alpha-1}}{\tr \sigma^{\alpha-1}} $ and $\frac{\rho^{\alpha}}{\tr \rho^{\alpha}}$, respectively. Then 
  \begin{align}
      D_{\alpha}(\rho||\sigma)<\infty \Leftrightarrow V_{\sigma(\alpha-1)}>V_{\rho(\alpha)}.
  \end{align}
\end{prop}
\begin{proof} 
Follows immediately from Remarks \ref{rmk:main-thm}  and \ref{rmk:V_inequality}.
\end{proof}
\subsection{Displaced Thermal States}\label{sec:displaced-thermal}
In the following theorem we verify a conjecture of Seshadreesan, Lami and Wilde \cite{Ses-Lam-Wil-2018} for the special case of displaced faithful thermal states.
\begin{thm}\label{thm:converse-SLW}
     Let $\rho$ and $\sigma$ be displaced faithful thermal states given by  \begin{align*}
        \rho&= W(\bu_1)\gamma(\br)W(\bu_1)^{\dagger},\\
        \sigma& = W(\bu_2)\gamma(\bs)W(\bu_2)^{\dagger} 
    \end{align*}
where $\bu_1 , \bu_2 \in \BC^{n}$, $\gamma(\br) =\otimes_{i=1}^n \gamma(r_i)$,  and $\gamma(\bs)=\otimes_{j=1}^n\gamma(s_j)$. For $\alpha>1$, let $V_{\sigma(\alpha-1)}$ and $V_{\rho(\alpha)}$ denote the covariance matrices of the gaussian states $\frac{\sigma^{\alpha-1}}{\tr \sigma^{\alpha-1}} $ and $\frac{\rho^{\alpha}}{\tr \rho^{\alpha}}$, respectively. Then we have \begin{align}
      D_{\alpha}(\rho||\sigma)<\infty \Rightarrow V_{\sigma(\alpha-1)}>V_{\rho(\alpha)}.
  \end{align}
\end{thm}
\begin{proof}
By Proposition \ref{prop:d-alpha-displaced-thermal}, $D_{\alpha}(\rho||\sigma)<\infty$ implies that \[\sum\limits_{\bk, \bm{\ell}\in \BZ_{\geq 0}^n }\lambda(\bk,\br)^\alpha\lambda(\bm{\ell},\bs)^{(1-\alpha)}\abs{\braket{W(\bu)\bk}{\bm{\ell}}}^2<\infty,\] where $\bu = \bu_1 -\bu_2$. Since the series of positive numbers above is convergent, we have in particular, the diagonal series \[\sum\limits_{\bk \in \BZ_{\geq 0}^n }\lambda(\bk,\br)^\alpha\lambda(\bm{\bk},\bs)^{(1-\alpha)}\abs{\braket{W(\bu)\bk}{\bm{\bk}}}^2<\infty.\] By the definition of $\lambda(\bk,\br)$ in (\ref{beq:5}), this is now same as \begin{align}\label{eq:intermediate}
    \sum\limits_{\bk \in \BZ_{\geq 0}^n }e^{-\bk\cdot (\alpha\br  +(1-\alpha)\bs)}\abs{\braket{W(\bu)\bk}{\bm{\bk}}}^2<\infty.
\end{align} Note that if $\bu = (u_1, u_2,\dots, u_n)$ and $\bk = (k_1, k_2, \dots, k_n)$,  then \begin{align*}
    \braket{W(\bu)\bk}{\bm{\bk}} &=  \braket{W(u_1)\otimes W(u_2)\otimes \cdots \otimes W(u_n)k_1\otimes k_2\otimes \cdots \otimes k_n}{k_1\otimes k_2\otimes \cdots \otimes k_n}\\
    &=\Pi_{j=1}^n \braket{W(u_j)k_j}{k_j}
\end{align*} Substituting back in \eqref{eq:intermediate}, we see that $D_\alpha(\rho||\sigma)<\infty$ implies that \[ \sum\limits_{(k_1,\dots, k_n) \in \BZ_{\geq 0}^n}e^{-\sum_j k_j (\alpha r_j  +(1-\alpha)s_j)}\prod_{j=1}^n \abs{\braket{W(u_j)k_j}{k_j}}^2<\infty.\]  Then we have \begin{align*}
   \prod_{i=1}^n \sum\limits_{{k_i} \in \BZ_{\geq0}}e^{- {k_i} (\alpha r_{i}  +(1-\alpha)s_{i})} \abs{\braket{W(u_{i})k_{i}}{k_{i}}}^2<\infty.
\end{align*}
By Proposition \ref{weyl-diagonal}, we know there exists $C_i>0$ such that 
\begin{align*}
|\langle W(u_{i}) k_{i}| k_{i}\rangle|^2 \geq\frac{C_i}{k_{i}^{\frac{3}{4}}}, 
\end{align*}
for infinitely many $k_i \in \mathbb{N}$, $\quad i=1,2,\dots, n$. Therefore we have \[\sum\limits_{{k_i} \in \BZ_{\geq0}}e^{- {k_i} (\alpha r_{i}  +(1-\alpha)s_{i})}\frac{C_i}{k_{i}^{\frac{3}{4}}} <\infty, \quad i=1,2,\dots,n.\] Therefore, \[\alpha r_{i}  +(1-\alpha)s_{i}>0, \quad i=1,2,\dots, n.\]
Now the equivalences in \eqref{eq:proof} show that the last displayed equality is equivalent to $V_{\sigma(\alpha-1)}>V_{\rho(\alpha)}$ because the covariance matrix of $\frac{\gamma(\br)^\alpha}{\tr \gamma(\br)^\alpha}$ is same as that of $\frac{\rho^\alpha}{\tr \rho^\alpha}$, (since the operations of raising to a positive power and displacing commute with each other).
\end{proof}
As we discussed earlier, Seshadreesan, Lami and Wilde \cite[Theorem 19]{Ses-Lam-Wil-2018}
proved that, \[V_{\sigma(\alpha-1)}>V_{\rho(\alpha)}\Rightarrow D_{\alpha}(\rho||\sigma)<\infty.\]
Combining this with our previous theorem, we obtain the following result.
\begin{cor}\label{cor:slw}
    Under the assumptions of Theorem \ref{thm:converse-SLW}, we have \begin{align}
      D_{\alpha}(\rho||\sigma)<\infty \Leftrightarrow V_{\sigma(\alpha-1)}>V_{\rho(\alpha)}.
  \end{align}
\end{cor}
Now we present the main result of our article which provides a precise range of $\alpha$ such that $D_\alpha(\rho||\sigma)<\infty$, for a given pair of displaced faithful  thermal states $\rho$ and $\sigma$. Note that Theorem \ref{thm:not-necessarily-faithful} does not have the assumption of faithfulness but also does not allow displacements. The next theorem allows displacements but requires faithfulness.
\begin{thm}\label{thm:main-faithful}
 Let $\rho$ and $\sigma$ be displaced faithful thermal states given by  \begin{align*}
        \rho&= W(\bu_1)\gamma(\br)W(\bu_1)^{\dagger},\\
        \sigma& = W(\bu_2)\gamma(\bs)W(\bu_2)^{\dagger} 
    \end{align*}
where $\bu_1 , \bu_2 \in \BC^{n}$, $\gamma(\br) =\otimes_{i=1}^n \gamma(r_i)$,  and $\gamma(\bs)=\otimes_{j=1}^n\gamma(s_j)$.  Then we have the following: \begin{align}
      D_{\alpha}(\rho||\sigma)<\infty \Leftrightarrow \alpha < \min \left\{ \frac{s_j}{s_j-r_j}: j \in \{ 1, \ldots , n \} \text{ such that } r_j<s_j \right\}, 
  \end{align} 
  where we adopt the convention that the minimum of an empty set is equal to infinity.
\end{thm}
\begin{proof} By Corollary \ref{cor:slw} $ D_{\alpha}(\rho||\sigma)<\infty \Leftrightarrow V_{\sigma(\alpha-1)}>V_{\rho(\alpha)}$. Since the covariance matrix of $\frac{\gamma(\br)^\alpha}{\tr \gamma(\br)^\alpha}$ is the same as that of $\frac{\rho^\alpha}{\tr \rho^\alpha}$, Lemma \ref{lem:slw}  completes the proof.
\end{proof} 
\appendix
\section{Appendix}
\subsection*{Single Mode Weyl Unitary Operators}
The classical Laguerre polynomials play an important role in our analysis of Weyl unitary operators in $1$-mode.
The Laguerre polynomials are defined as \begin{align}\label{eq:laguerre}
    L_j(x)=\sum_{k=0}^j
\binom{j}{k}\frac{(-1)^k}{k !}x^k, \quad j=0,1,2,\dots.
\end{align}
The following proposition is available in \cite{jencova-petz-pitrik-2010} with a different proof.
\begin{prop}\label{prop:weyl-leguerre}\cite{jencova-petz-pitrik-2010}
Let $L_j$ denote the $j$-th Laguerre polynomial, $j\in \BZ_{\geq 0}$. Then the $1$-mode Weyl unitary operator at $u\in \BC$, satisfy  \[\mel{j}{W(u)}{j} = \exp{-\frac{1}{2}\abs{u}^2}L_j(\abs{u}^2).\] 
\end{prop}
\begin{proof} 
For clarity of writing in the proof, we denote the particle basis vector $\ket{j}$ of $\Gamma(\BC)$  as $f_j$.   Notice that $f_i\otimes f_j = f_{i+j}$ for $i,j\geq 0$.  Note also that in the one-mode case $u=uf_1$ for all $u\in \BC$.  So, we have \begin{align*}
    (u+t f_1)^{\otimes^\ell}=(u+t)^\ell \left(f_1\right)^{\otimes^\ell} = \sum_{m=0}^{\ell}\binom{\ell}{m}u^{\ell-m}t^mf_{\ell}, \quad \forall u\in \BC, t\in \BR.
\end{align*}
Now, by the definition of exponential vectors and the action of Weyl unitary operators on the exponential vectors \eqref{eq:weyl-exp}, 
    \begin{align*}
\sum_{j=0}^{\infty} t^j W(u) \frac{f_j}{\sqrt{j !}}&=W(u) \ket{e(t f_1)}\\&=\exp \left\{-\frac{1}{2}|u|^2- \bar{u}t \right\} \ket{e(u+t f_1)}, \quad u\in \BC, t\in \BR\\
&= \exp \left\{-\frac{1}{2}|u|^2- \bar{u}t \right\} \sum_{\ell=0}^\infty \left(  \frac{\sum_{m=0}^{\ell}\binom{\ell}{m}u^{\ell-m}t^m}{\sqrt{\ell!}}\right)f_{\ell}\\
& =\exp \left\{-\frac{1}{2}|u|^2\right\}\left(\sum_{k=0}^{\infty}(-1)^k \frac{(\bar{u}t)^k}{k !} \right)  \sum_{m=0}^{\infty}\sum_{\ell=m}^{\infty}\frac{\binom{\ell}{m} u^{{\ell-m}} t^{m} }{\sqrt{\ell !}}f_{\ell}, \quad\forall u\in \BC, t\in \BR.
\end{align*}
Identifying the coefficient of $t^j$ on both sides of the equation above and writing $m+k=j$, we have
\begin{align}\label{eq:wu-v-tensor}
W(u) f_j=\sqrt{j !} \exp \left\{-\frac{1}{2}|u|^2\right\} \left(\sum_{\{m,k \geq 0\mid m+k = j\}} (-1)^k \frac{ \bar{u}^k}{k !}\right)\left(\sum_{\ell=m}^{\infty} \frac{\binom{\ell}{m} u^{{\ell-m}} }{\sqrt{\ell !}}f_{\ell}\right).
\end{align}
Hence, from \eqref{eq:wu-v-tensor}, we get
\begin{align*}
\left\langle f_j \mid W(u) f_j\right\rangle &= \sqrt{j !} \exp \left\{-\frac{1}{2}|u|^2\right\}\left(\sum_{\{m,k \geq 0\mid m+k = j\}} (-1)^k \frac{ \bar{u}^k}{k !}\right) \sum_{\ell=m}^{\infty}\binom{\ell}{m} \frac{u^{\ell-m}}{\sqrt{\ell !}}\left\langle f_j | f_{\ell}\right\rangle\\
&= \sqrt{j !} \exp \left\{-\frac{1}{2}|u|^2\right\}\left(\sum_{\{m,k \geq 0\mid m+k = j\}} (-1)^k \frac{ \bar{u}^k}{k !}\right) \sum_{\ell=m}^{\infty}\binom{\ell}{m} \frac{u^{\ell-m}}{\sqrt{\ell !}}\delta_{j,\ell}\\
&= \sqrt{j !} \exp \left\{-\frac{1}{2}|u|^2\right\}\left(\sum_{\{m,k \geq 0\mid m+k = j\}} (-1)^k \frac{ \bar{u}^k}{k !}\right) \binom{j}{m} \frac{u^{j-m}}{\sqrt{j !}}\\
&= \sqrt{j !} \exp \left\{-\frac{1}{2}|u|^2\right\}\left(\sum_{\{m,k \geq 0\mid m+k = j\}} (-1)^k \frac{ \bar{u}^k}{k !} \binom{j}{j-k} \frac{u^{k}}{\sqrt{j !}}\right)\\
& =\sqrt{j !} \exp \left\{-\frac{1}{2}|u|^2\right\} \sum_{k=0}^j \frac{(-1)^k\bar{u}^k}{k !}\binom{j}{k}\frac{u^k}{\sqrt{j !}} \\
& =\exp \left\{-\frac{1}{2}|u|^2\right\} \sum_{k=0}^j(-1)^k\binom{j}{k} \frac{\left(|u|^2\right)^k}{k !}.
\end{align*}
Thus from \eqref{eq:laguerre}, we get 
\[\left\langle {j} \mid W(u) \mid { j}\right\rangle =\exp \left\{-\frac{1}{2}|u|^2\right\} \sum_{k=0}^j
\binom{j}{k}\frac{(-1)^k}{k !}(|u|^2)^k =\exp{-\frac{1}{2}\abs{u}^2}L_j(\abs{u}^2). \]
\end{proof}

\begin{lem}\label{lem:analysis} For every $u\in \BC$, there is an infinite sequence of positive integers $j$ such that
\[|\sin (2 \sqrt{j}|u|+\pi / 4)| \geqslant \frac{1}{\sqrt{2}}.\]
\end{lem}
\begin{proof}
    Observe that 
\begin{align*}
\begin{aligned}
& |\sin (2 \sqrt{j}\abs{u}+\pi / 4)| \geqslant \frac{1}{\sqrt{2}} \\&\phantom{............} \Leftrightarrow \text{there exists } \, m \in \BZ \text{ such that }  2 \sqrt{j}|u|+\frac{\pi}{4} \in\left[m \pi+\frac{\pi}{4}, m \pi+ \frac{3\pi}{4}\right].
\end{aligned}
\end{align*}
Notice that
\begin{align*}
\begin{aligned}
& m \pi+\frac{\pi}{4} \leqslant 2 \sqrt{j}|u|+\frac{\pi}{4} \leqslant m \pi+\frac{3 \pi}{4} 
\Leftrightarrow \left(\frac{m \pi}{2|u|}\right)^2 \leqslant j \leqslant\left(\frac{m \pi}{2|u|}\right)^2+\frac{m \pi^2}{4|u|^2}+\frac{\pi^2}{16| u|^2}.
\end{aligned}
\end{align*}
Since both $\frac{m \pi^2}{4|u|^2}+\frac{\pi^2}{16| u|^2} \rightarrow \infty$  and $\left(\frac{m \pi}{2|u|}\right)^2 \rightarrow \infty$ as $m\rightarrow \infty$,
we can find an infinite sequence of positive integers $m$ such that the intervals\[I_m=\left[\left(\frac{m \pi}{2|u|}\right)^2,\left(\frac{m \pi}{2|u|}\right)^2+\frac{m \pi^2}{4|u|^2}+\frac{\pi^2}{16| u|^2}\right]\]
 are pairwise disjoint and have lengths bigger than one. Therefore, for every such interval we can choose a positive integer $j$ in $I_m$. Thus we have an infinite sequence of $j$'s such that
\[|\sin (2 \sqrt{j}|u|+\pi / 4)| \geqslant \frac{1}{\sqrt{2}}.\]
\end{proof}
\begin{prop}\label{weyl-diagonal}
    For every  $u\in \BC$, there exists $C>0$ such that \[\abs{\mel{j}{W(u)}{j}}\geq \frac{C}{j^{\frac{3}{8}}},\] for infinitely many $j\in \BZ_{\geq 0}$.
\end{prop}
\begin{proof}
    By Proposition \ref{prop:weyl-leguerre} it is enough to show that for every $u\in \BC$, there exists $C>0$ such that \[\abs{L_j(\abs{u}^2)}\geq \frac{C}{j^{\frac{3}{8}}}, \]
    for infinitely many $j\in \BZ_{\geq 0}$. By  Fejér's formula \cite[Theorem 8.22.1]{szego1939orthogonal} (alternatively \cite{wiki-Laguerre} taking $\alpha =0$),  we have for large $j$'s \[L_j(x)=\frac{j^{-\frac{1}{4}}}{\sqrt{\pi}} \frac{e^{\frac{x}{2}}}{x^{\frac{1}{4}}} \sin \left(2 \sqrt{j x}+\frac{\pi}{4}\right)+O\left(j^{-\frac{3}{4}}\right), \quad \forall x >0.\]
    Taking $x = \abs{u}^2$, we get
\begin{align*}
    L_j(\abs{u}^2)&=\frac{e^{\frac{\abs{u}^2}{2}}}{\sqrt{\pi\abs{u}}} \frac{1}{j^{\frac{1}{4}}} \sin \left(2 \sqrt{j}\abs{u}+\frac{\pi}{4}\right)+O\left(j^{-\frac{3}{4}}\right)
\end{align*}
Using Lemma \ref{lem:analysis} we see that there exists positive constants $C_1$ and $C_2$ (depending only on $u$) such that for infinitely many $j$'s
\begin{align*}
    L_j(\abs{u}^2) \geq \frac{e^{\frac{\abs{u}^2}{2}}}{\sqrt{2\pi\abs{u}}} \frac{1}{j^{\frac{1}{4}}} - \frac{C_2}{j^{\frac{3}{4}}} = \frac{C_1}{j^{\frac{1}{4}}}-\frac{C_2}{j^{\frac{3}{4}}} 
\end{align*}
Thus for every $\epsilon >0$, there exists $C>0$ such that 
\[L_j(\abs{u}^2) \geq \frac{C}{j^{\frac{1}{4}+\epsilon}},\]
for infinitely many $j$'s. We conclude the proof by choosing $\epsilon = 1/8.$
\end{proof}

\section*{Acknowledgement}The authors are grateful to  the referees whose suggestions greatly improved the article.   The second author wishes to acknowledge the Army Research Office MURI award `Theory and Engineering of Large-Scale Distributed Entanglement Quantum Network Science – QNS' awarded under grant number W911NF2110325 – ARO MURI. The second author also thanks the Fulbright Scholar Program and the United States-India Educational Foundation for providing funding to conduct part of this research through a Fulbright-Nehru Postdoctoral Fellowship (Award No. 2594/FNPDR/2020).

\noindent
\textbf{Competing Interests:} The authors do not have any competing interests.

\noindent
\textbf{Data Availability:}
 The data that supports the findings of this study are available within the article.

\bibliographystyle{IEEEtran}
\bibliography{bibliography}
\end{document}